\documentclass[11pt]{article}
\usepackage[a4paper]{geometry}
\usepackage{amsfonts, amsmath, amssymb, amsthm, graphicx, caption, authblk, multirow, makecell, framed, float, xcolor, enumitem, tikz, hyperref}
\theoremstyle{definition}

\newtheorem{lemma}{Lemma}

\theoremstyle{remark}

\newcommand*{\mybox}[1]{%
  \framebox{\raisebox{0cm}[0.5\baselineskip][0.05\baselineskip]{%
    \hbox to 0.1cm {\hss#1\hss}}}\hspace{0.05cm}}

\begin{document}
\title{How to Physically Verify a Rectangle in a Grid: A Physical ZKP for Shikaku}
\author[1]{Suthee Ruangwises\thanks{\texttt{ruangwises@gmail.com}}}
\author[1]{Toshiya Itoh\thanks{\texttt{titoh@c.titech.ac.jp}}}
\affil[1]{Department of Mathematical and Computing Science, Tokyo Institute of Technology, Tokyo, Japan}
\date{}
\maketitle

\begin{abstract}
Shikaku is a pencil puzzle consisting of a rectangular grid, with some cells containing a number. The player has to partition the grid into rectangles such that each rectangle contains exactly one number equal to the area of that rectangle. In this paper, we propose two physical zero-knowledge proof protocols for Shikaku using a deck of playing cards, which allow a prover to physically show that he/she knows a solution of the puzzle without revealing it. Most importantly, in our second protocol we develop a general technique to physically verify a rectangle-shaped area with a certain size in a rectangular grid, which can be used to verify other problems with similar constraints.

\end{abstract}

\section{Introduction}
\textit{Shikaku} is a pencil puzzle introduced by Nikoli, a Japanese publisher that developed many popular pencil puzzles such as Sudoku, Kakuro, and Slitherlink. The puzzle has become popular and many Shikaku mobile apps have been developed \cite{google}. A Shikaku puzzle consists of a rectangular grid of size $m \times n$, with some cells containing a number. The objective of this puzzle is to partition the grid into rectangles such that each rectangle contains exactly one number, which must be equal to the area of that rectangle (see Figure \ref{fig1}). Determining whether a given Shikaku puzzle has a solution is an NP-complete problem \cite{np}.

\begin{figure}
\centering
\begin{tikzpicture}
\draw[step=0.6cm,color={rgb:black,1;white,4}] (0,0) grid (4.2,4.2);

\node at (0.3,0.3) {3};
\node at (1.5,0.3) {6};
\node at (3.3,0.3) {4};
\node at (2.1,0.9) {3};
\node at (2.7,0.9) {2};
\node at (3.3,1.5) {3};
\node at (1.5,2.1) {2};
\node at (3.9,2.1) {4};
\node at (0.3,2.7) {2};
\node at (3.9,2.7) {2};
\node at (0.3,3.3) {2};
\node at (1.5,3.3) {8};
\node at (2.7,3.9) {4};
\node at (3.3,3.9) {2};
\node at (3.9,3.9) {2};
\end{tikzpicture}
\hspace{1.2cm}
\begin{tikzpicture}
\draw[step=0.6cm,color={rgb:black,1;white,4}] (0,0) grid (4.2,4.2);

\draw[line width=0.5mm] (0,0) -- (0,4.2);
\draw[line width=0.5mm] (0.6,0) -- (0.6,4.2);
\draw[line width=0.5mm] (1.8,0) -- (1.8,2.4);
\draw[line width=0.5mm] (2.4,0) -- (2.4,1.8);
\draw[line width=0.5mm] (3,0) -- (3,1.2);
\draw[line width=0.5mm] (3,2.4) -- (3,4.2);
\draw[line width=0.5mm] (3.6,3) -- (3.6,4.2);
\draw[line width=0.5mm] (4.2,0) -- (4.2,4.2);

\draw[line width=0.5mm] (0,0) -- (4.2,0);
\draw[line width=0.5mm] (2.4,1.2) -- (4.2,1.2);
\draw[line width=0.5mm] (0,1.8) -- (4.2,1.8);
\draw[line width=0.5mm] (0.6,2.4) -- (4.2,2.4);
\draw[line width=0.5mm] (0,3) -- (0.6,3);
\draw[line width=0.5mm] (3,3) -- (4.2,3);
\draw[line width=0.5mm] (0.6,3.6) -- (3,3.6);
\draw[line width=0.5mm] (0,4.2) -- (4.2,4.2);

\node at (0.3,0.3) {3};
\node at (1.5,0.3) {6};
\node at (3.3,0.3) {4};
\node at (2.1,0.9) {3};
\node at (2.7,0.9) {2};
\node at (3.3,1.5) {3};
\node at (1.5,2.1) {2};
\node at (3.9,2.1) {4};
\node at (0.3,2.7) {2};
\node at (3.9,2.7) {2};
\node at (0.3,3.3) {2};
\node at (1.5,3.3) {8};
\node at (2.7,3.9) {4};
\node at (3.3,3.9) {2};
\node at (3.9,3.9) {2};
\end{tikzpicture}
\caption{An example of a $7 \times 7$ Shikaku puzzle (left) and its solution (right)}
\label{fig1}
\end{figure}

Suppose that Paimon, an expert in Shikaku, created a difficult Shikaku puzzle and challenged her friend Venti to solve it. After a while, Venti could not solve her puzzle and began to doubt whether the puzzle actually has a solution. Paimon wants to convince him that her puzzle indeed has a solution without revealing it (which would render the challenge pointless). To achieve this, Paimon needs a \textit{zero-knowledge proof (ZKP)}.

\subsection{Zero-Knowledge Proof}
First introduced in 1989 by Goldwasser et al. \cite{zkp0}, a ZKP is an interactive protocol between a prover $P$ and a verifier $V$. Both $P$ and $V$ are given a computational problem $x$, but only $P$ knows a solution $w$ of $x$. A ZKP enables $P$ to convince $V$ that he/she knows $w$ without revealing any information about $w$. It must satisfy the following three properties.
\begin{enumerate}
	\item \textbf{Completeness:} If $P$ knows $w$, then $V$ accepts with high probability. (In this paper, we consider only the \textit{perfect completeness} property where $V$ always accepts.)
	\item \textbf{Soundness:} If $P$ does not know $w$, then $V$ rejects with high probability. (In this paper, we consider only the \textit{perfect soundness} property where $V$ always rejects.)
	\item \textbf{Zero-knowledge:} $V$ learns nothing about $w$. Formally, there exists a probabilistic polynomial time algorithm $S$ (called a \textit{simulator}), not knowing $w$ but having an access to $V$, such that the outputs of $S$ follow the same probability distribution as the ones of the real protocol.
\end{enumerate}

As there exists a ZKP for every NP problem \cite{zkp}, one can construct a computational ZKP for Shikaku. However, such construction requires cryptographic primitives and thus is not intuitive or practical.

Instead, many results so far aimed to develop physical ZKP protocols using a deck of playing cards. These card-based protocols have benefits that they use only portable objects found in everyday life and do not require computers. They also allow external observers to verify that the prover truthfully executes the protocol (which is often a challenging task for digital protocols). In addition, these protocols have great didactic values to teach the concept of a ZKP to non-experts.

\subsection{Related Work}
Card-based ZKP protocols for many other popular pencil puzzles have been developed, including Sudoku \cite{sudoku0,sudoku2,sudoku}, Nonogram \cite{nonogram,nonogram2}, Akari \cite{akari}, Takuzu \cite{akari,takuzu}, Kakuro \cite{akari,kakuro}, KenKen \cite{akari}, Makaro \cite{makaro}, Norinori \cite{norinori}, Slitherlink \cite{slitherlink}, Juosan \cite{takuzu}, Numberlink \cite{numberlink}, Suguru \cite{suguru}, Ripple Effect \cite{ripple}, Nurikabe \cite{nurikabe}, Hitori \cite{nurikabe}, Cryptarithmetic \cite{crypta}, and Bridges \cite{bridges}.

In a recent work of Robert et al. \cite{nurikabe}, the authors posed an open problem to extend the idea of their protocol to verify a solution of Shikaku or other puzzles that require to draw rectangles with certain sizes in a grid.

\subsection{Our Contribution}
In this paper, we answer the open problem posed by Robert et al. \cite{nurikabe} by developing two card-based ZKP protocols with perfect completeness and soundness for Shikaku: a brute force protocol and a more elegant, intuitive \textit{flooding protocol}. The two protocols use $\Theta(m^2n^2)$ cards and $\Theta(mn)$ cards, respectively.

Most importantly, in the flooding protocol we develop a general technique to physically verify a rectangle-shaped area with a certain size in a rectangular grid, which can be used to verify other problems with similar constraints.

\section{First Attempt: Brute Force Protocol}
Every card used in this paper has an integer on the front side. All cards have indistinguishable back sides denoted by \mybox{?}.

Let $(x,y)$ denote a cell located in the $x$-th topmost row and $y$-th leftmost column of the Shikaku grid. Let $p_2,p_3,...,p_{k+1}$ be the $k$ numbers written on the grid\footnote{We intentionally start the indices at 2 so that our second protocol, which will be introduced later, will be easier to understand.}, with each number $p_i$ in a cell $(x_i,y_i)$. Note that we must have $p_2+p_3+...+p_{k+1} = mn$.

Suppose that in $P$'s solution, the grid is divided into $k$ rectangles $Z_2,Z_3,...,Z_{k+1}$ such that each $Z_i$ contains the number $p_i$. Each rectangle $Z_i$ is represented by its top-left and bottom-right corner cells $(a_i,b_i)$ and $(a'_i,b'_i)$, respectively. To verify that the solution is correct, it is sufficient to show that
\begin{enumerate}
	\item $a_i \leq x_i \leq a'_i$ and $b_i \leq y_i \leq b'_i$ (a cell with the number $p_i$ is inside $Z_i$) for every $i \in \{2,3,...,k+1\}$,
	\item $(a'_i-a_i+1)(b'_i-b_i+1) = p_i$ (the area of $Z_i$ is equal to $p_i$) for every $i \in \{2,3,...,k+1\}$, and
	\item $a'_i < a_j$ or $a'_j < a_i$ or $b'_i < b_j$ or $b'_j < b_i$ ($Z_i$ and $Z_j$ do not overlap) for every distinct $i,j \in \{2,3,...,k+1\}$.
\end{enumerate}

These three conditions can be verified by applying the combination of the copy, addition, multiplication, and equality protocols \cite{bridges}, and a protocol to compare two numbers \cite{makaro}.

This protocol, however, involves a lot of messy calculations and thus has lost its didactic values as it becomes more computational and less intuitive. Moreover, it requires up to $\Theta(m^2n^2)$ cards (as we have to multiply integers in modulo $mn$)\footnote{In this protocol, an integer $x$ in modulo $mn$ is encoded by a sequence of $mn$ consecutive cards, with all of them being \mybox{0}s except the $(x+1)$-th leftmost card being a \mybox{1}.}, which is far too many to be practical. Instead, we are looking for an elegant and intuitive protocol that uses a reasonable number of cards.

\section{Verifying an Area of Connected Cells}
In a recent work, Robert et al. \cite{nurikabe} developed a \textit{sea formation protocol} that allows the prover $P$ to convince the verifier $V$ that a given area in a grid consists of $t$ cells that are connected to each other horizontally or vertically. We will first show the necessary subprotocols and then explain the sea formation protocol.

\subsection{Pile-Shifting Shuffle}
Given a $p \times q$ matrix of cards, a \textit{pile-shifting shuffle} rearranges the columns of the matrix by a random cyclic shift, i.e. shifts the columns cyclically to the right by $x$ columns for a uniformly random $x \in \mathbb{Z}/q\mathbb{Z}$, unknown to all parties.

The pile-shifting shuffle was developed by Shinagawa et al. \cite{polygon}. It can be performed in real world by putting the cards in each column into an envelope and then taking turns to apply \textit{Hindu cuts} (taking several envelopes from the bottom and putting them on the top) to the sequence of envelopes \cite{hindu}.

\subsection{Chosen Cut Protocol} \label{chosen}
Given a sequence of $q$ face-down cards $C = (c_1,c_2,...,c_q)$, a \textit{chosen cut protocol} for $q$ cards allows $P$ to select a card $c_i$ he/she wants (to use in other operations) without revealing $i$ to $V$. This protocol also reverts the sequence $C$ back to its original state after $P$ finishes using $c_i$. It was developed by Koch and Walzer \cite{koch}.

\begin{figure}[h]
\centering
\begin{tikzpicture}
\node at (0.0,2.4) {\mybox{?}};
\node at (0.6,2.4) {\mybox{?}};
\node at (1.2,2.4) {...};
\node at (1.8,2.4) {\mybox{?}};
\node at (2.4,2.4) {\mybox{?}};
\node at (3.0,2.4) {\mybox{?}};
\node at (3.6,2.4) {...};
\node at (4.2,2.4) {\mybox{?}};

\node at (0.0,2) {$c_1$};
\node at (0.6,2) {$c_2$};
\node at (1.8,2) {$c_{i-1}$};
\node at (2.4,2) {$c_i$};
\node at (3.0,2) {$c_{i+1}$};
\node at (4.2,2) {$c_q$};

\node at (0.0,1.4) {\mybox{?}};
\node at (0.6,1.4) {\mybox{?}};
\node at (1.2,1.4) {...};
\node at (1.8,1.4) {\mybox{?}};
\node at (2.4,1.4) {\mybox{?}};
\node at (3.0,1.4) {\mybox{?}};
\node at (3.6,1.4) {...};
\node at (4.2,1.4) {\mybox{?}};

\node at (0.0,1) {0};
\node at (0.6,1) {0};
\node at (1.8,1) {0};
\node at (2.4,1) {1};
\node at (3.0,1) {0};
\node at (4.2,1) {0};

\node at (0.0,0.4) {\mybox{?}};
\node at (0.6,0.4) {\mybox{?}};
\node at (1.2,0.4) {...};
\node at (1.8,0.4) {\mybox{?}};
\node at (2.4,0.4) {\mybox{?}};
\node at (3.0,0.4) {\mybox{?}};
\node at (3.6,0.4) {...};
\node at (4.2,0.4) {\mybox{?}};

\node at (0.0,0) {1};
\node at (0.6,0) {0};
\node at (1.8,0) {0};
\node at (2.4,0) {0};
\node at (3.0,0) {0};
\node at (4.2,0) {0};
\end{tikzpicture}
\caption{A $3 \times q$ matrix $M$ constructed in Step 1 of the chosen cut protocol}
\label{fig2}
\end{figure}

\begin{enumerate}
	\item Construct the following $3 \times q$ matrix $M$ (see Figure \ref{fig2}).
	\begin{enumerate}
		\item In Row 1, publicly place the sequence $C$.
		\item In Row 2, secretly place a face-down \mybox{1} at Column $i$ and a face-down \mybox{0} at each other column.
		\item In Row 3, secretly place a face-down \mybox{1} at Column 1 and a face-down \mybox{0} at each other column.
	\end{enumerate}
	\item Apply the pile-shifting shuffle to $M$.
	\item Turn over all cards in Row 2. Locate the position of the only \mybox{1}. A card in Row 1 directly above that \mybox{1} will be the card $c_i$ as desired.
	\item After we finish using $c_i$ in other operations, place $c_i$ back into $M$ at the same position.
	\item Turn over all face-up cards in Row 2 and apply the pile-shifting shuffle to $M$ again.
	\item Turn over all cards in Row 3. Locate the position of the only \mybox{1}. Shift the columns of $M$ cyclically such that this \mybox{1} moves to Column 1. This reverts $M$ back to its original state.
\end{enumerate}

Note that Steps 3 and 6 of this protocol guarantee that the cards in Row 2 and Row 3 are in a correct format (each row having one \mybox{1} and $q-1$ \hbox{\mybox{0}s}).

\subsection{Sea Formation Protocol}
First, publicly place a face-down \mybox{0} on every cell in the Shikaku grid. To handle the case where a selected cell is on the edge of the grid, we publicly place face-down ``dummy cards'' $\mybox{-1}$s around the grid. We now have an $(m+2) \times (n+2)$ matrix of cards (see Figure \ref{fig3}).

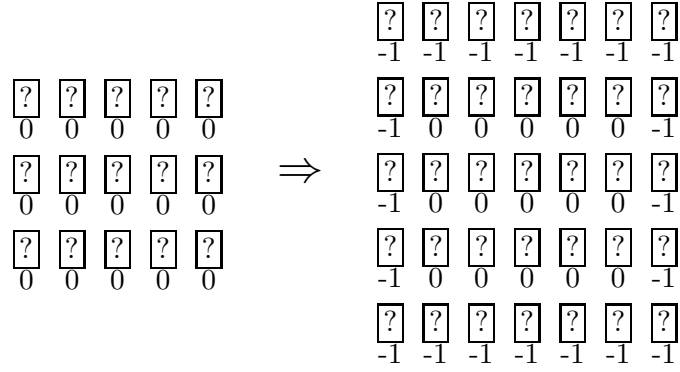
\begin{figure}[h]
\centering
\begin{tikzpicture}
\node at (0.6,3.4) {\mybox{?}};
\node at (1.2,3.4) {\mybox{?}};
\node at (1.8,3.4) {\mybox{?}};
\node at (2.4,3.4) {\mybox{?}};
\node at (3.0,3.4) {\mybox{?}};

\node at (0.6,3) {0};
\node at (1.2,3) {0};
\node at (1.8,3) {0};
\node at (2.4,3) {0};
\node at (3.0,3) {0};

\node at (0.6,2.4) {\mybox{?}};
\node at (1.2,2.4) {\mybox{?}};
\node at (1.8,2.4) {\mybox{?}};
\node at (2.4,2.4) {\mybox{?}};
\node at (3.0,2.4) {\mybox{?}};

\node at (0.6,2) {0};
\node at (1.2,2) {0};
\node at (1.8,2) {0};
\node at (2.4,2) {0};
\node at (3.0,2) {0};

\node at (0.6,1.4) {\mybox{?}};
\node at (1.2,1.4) {\mybox{?}};
\node at (1.8,1.4) {\mybox{?}};
\node at (2.4,1.4) {\mybox{?}};
\node at (3.0,1.4) {\mybox{?}};

\node at (0.6,1) {0};
\node at (1.2,1) {0};
\node at (1.8,1) {0};
\node at (2.4,1) {0};
\node at (3.0,1) {0};

\node at (3.0,-0.1) {};

\node at (4.2,2.4) {\LARGE{$\Rightarrow$}};
\end{tikzpicture}
\hspace{0.2cm}
\begin{tikzpicture}
\node at (0.0,4.4) {\mybox{?}};
\node at (0.6,4.4) {\mybox{?}};
\node at (1.2,4.4) {\mybox{?}};
\node at (1.8,4.4) {\mybox{?}};
\node at (2.4,4.4) {\mybox{?}};
\node at (3.0,4.4) {\mybox{?}};
\node at (3.6,4.4) {\mybox{?}};

\node at (0.0,4) {-1};
\node at (0.6,4) {-1};
\node at (1.2,4) {-1};
\node at (1.8,4) {-1};
\node at (2.4,4) {-1};
\node at (3.0,4) {-1};
\node at (3.6,4) {-1};

\node at (0.0,3.4) {\mybox{?}};
\node at (0.6,3.4) {\mybox{?}};
\node at (1.2,3.4) {\mybox{?}};
\node at (1.8,3.4) {\mybox{?}};
\node at (2.4,3.4) {\mybox{?}};
\node at (3.0,3.4) {\mybox{?}};
\node at (3.6,3.4) {\mybox{?}};

\node at (0.0,3) {-1};
\node at (0.6,3) {0};
\node at (1.2,3) {0};
\node at (1.8,3) {0};
\node at (2.4,3) {0};
\node at (3.0,3) {0};
\node at (3.6,3) {-1};

\node at (0.0,2.4) {\mybox{?}};
\node at (0.6,2.4) {\mybox{?}};
\node at (1.2,2.4) {\mybox{?}};
\node at (1.8,2.4) {\mybox{?}};
\node at (2.4,2.4) {\mybox{?}};
\node at (3.0,2.4) {\mybox{?}};
\node at (3.6,2.4) {\mybox{?}};

\node at (0.0,2) {-1};
\node at (0.6,2) {0};
\node at (1.2,2) {0};
\node at (1.8,2) {0};
\node at (2.4,2) {0};
\node at (3.0,2) {0};
\node at (3.6,2) {-1};

\node at (0.0,1.4) {\mybox{?}};
\node at (0.6,1.4) {\mybox{?}};
\node at (1.2,1.4) {\mybox{?}};
\node at (1.8,1.4) {\mybox{?}};
\node at (2.4,1.4) {\mybox{?}};
\node at (3.0,1.4) {\mybox{?}};
\node at (3.6,1.4) {\mybox{?}};

\node at (0.0,1) {-1};
\node at (0.6,1) {0};
\node at (1.2,1) {0};
\node at (1.8,1) {0};
\node at (2.4,1) {0};
\node at (3.0,1) {0};
\node at (3.6,1) {-1};

\node at (0.0,0.4) {\mybox{?}};
\node at (0.6,0.4) {\mybox{?}};
\node at (1.2,0.4) {\mybox{?}};
\node at (1.8,0.4) {\mybox{?}};
\node at (2.4,0.4) {\mybox{?}};
\node at (3.0,0.4) {\mybox{?}};
\node at (3.6,0.4) {\mybox{?}};

\node at (0.0,0) {-1};
\node at (0.6,0) {-1};
\node at (1.2,0) {-1};
\node at (1.8,0) {-1};
\node at (2.4,0) {-1};
\node at (3.0,0) {-1};
\node at (3.6,0) {-1};
\end{tikzpicture}
\caption{The way we place cards on a $3 \times 5$ Shikaku grid during the setup of the sea formation protocol}
\label{fig3}
\end{figure}

Start at the top-left corner of the matrix and pick all cards in the order from left to right in Row 1, then from left to right in Row 2, and so on. Arrange them into a single sequence $D=(d_1,d_2,...,d_{(m+2)(n+2)})$. Note that we know exactly where the four neighbors of any given card are. Namely, the cards on the neighbor to the left, right, top, and bottom of a cell containing $d_i$ are $d_{i-1}$, $d_{i+1}$, $d_{i-n-2}$, and $d_{i+n+2}$, respectively.

The sea formation protocol to verify a connected area with size $t$ works as follows.

\begin{enumerate}
	\item $P$ applies the chosen cut protocol for $(m+2)(n+2)$ cards to select a \mybox{0} that he/she wants to replace.
	\item $P$ reveals the selected card to $V$ that it is a \mybox{0} (otherwise $V$ rejects) and then replaces it with a \mybox{1}.
	\item $P$ repeatedly performs the following steps for $t-1$ iterations.
	\begin{enumerate}
		\item $P$ applies the chosen cut protocol for $(m+2)(n+2)$ cards to select a \mybox{1} he/she wants.
		\item $P$ reveals the selected card to $V$ that it is a \mybox{1} (otherwise $V$ rejects).
		\item $P$ picks the four neighbors of the selected card and applies the chosen cut protocol for four cards to select one of the four neighbors, which is a \mybox{0} that he/she wants to replace.
		\item $P$ reveals the selected neighbor to $V$ that it is a \mybox{0} (otherwise $V$ rejects) and then replaces it with a \mybox{1}.
	\end{enumerate}
\end{enumerate}

We can see that in each iteration, the ``sea'' of \mybox{1}s expands by one cell, while all \mybox{1}s remain connected to each other. Therefore, after $t-1$ steps, $V$ is convinced that there is an area of $t$ \mybox{1}s in the grid that are connected to each other.

\section{Idea to Verify a Rectangle-Shaped Area}
The sea formation protocol, however, does not say anything about the shape of the area. By extending the idea of the sea formation protocol, we propose the following \textit{flooding protocol}, which allows $P$ to convince $V$ that the area is a rectangle with size $t$.

The idea is to always start at the top-left corner of the rectangle. At first, $P$ changes the card on the top-left corner cell of the rectangle from a \mybox{0} to a \mybox{1}. Similarly to the sea formation protocol, in each step $P$ selects a cell with a \mybox{1} and changes the card on one of its neighbor from a \mybox{0} to a \mybox{1}. However, the difference from the sea formation protocol is that $P$ can only select the neighbor to the right or to the bottom (but not to the left or to the top). We call this process a \textit{flood}, which starts at the top-left corner and goes downwards or rightwards in each step until it eventually fills the whole rectangle in $t-1$ steps.

To be more specific, at first the flood can only go downwards (i.e. $P$ can only select the neighbor to the bottom) to fill cells along the left edge of the rectangle. Then, right after it just filled all cells along the left edge, the flood suddenly changes direction and can only go rightwards (i.e. $P$ can only select the neighbor to the right) to fill the rest of the cells in the rectangle. In particular, $V$ must not know the exact time when the flood changes direction (otherwise $V$ will know the height of the rectangle).

The technique to achieve this ``one-time direction change'' is to let $P$ keep a secret variable $r$, which controls the direction of the flood (if $r=0$, then the flood goes downwards; if $r=1$, then the flood goes rightwards). At the beginning, $P$ shows $V$ that $r=0$. Before each step, $P$ secretly chooses whether to add 1 to $r$ or not, then shows $V$ that $r \neq 2$ (without revealing the actual value of $r$). This technique works because while $r=0$, $r$ can become either 0 or 1 in the next step, but once $r$ becomes 1, it must remain 1 forever (see a subprotocol in Section \ref{select} on how to make the selected neightbor depend on the value of $r$).

$P$ performs the above process for $t-1$ times to change all \mybox{0}s in the rectangle to \hbox{\mybox{1}s}. However, the protocol is not finished yet, as $V$ is not yet convinced that the area is a rectangle. In fact, $P$ has only shown that the area has a straight left edge; it may look like one of the shapes in Figure \ref{fig4}.

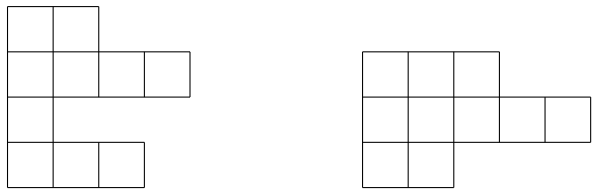
\begin{figure}[h]
\centering
\begin{tikzpicture}

\draw (0,0) -- (0,2.4);
\draw (0.6,0) -- (0.6,2.4);
\draw (1.2,0) -- (1.2,0.6);
\draw (1.2,1.2) -- (1.2,2.4);
\draw (1.8,0) -- (1.8,0.6);
\draw (1.8,1.2) -- (1.8,1.8);
\draw (2.4,1.2) -- (2.4,1.8);

\draw (0,0) -- (1.8,0);
\draw (0,0.6) -- (1.8,0.6);
\draw (0,1.2) -- (2.4,1.2);
\draw (0,1.8) -- (2.4,1.8);
\draw (0,2.4) -- (1.2,2.4);

\end{tikzpicture}
\hspace{2cm}
\begin{tikzpicture}

\draw (0,0) -- (0,1.8);
\draw (0.6,0) -- (0.6,1.8);
\draw (1.2,0) -- (1.2,1.8);
\draw (1.8,0.6) -- (1.8,1.8);
\draw (2.4,0.6) -- (2.4,1.2);
\draw (3,0.6) -- (3,1.2);

\draw (0,0) -- (1.2,0);
\draw (0,0.6) -- (3,0.6);
\draw (0,1.2) -- (3,1.2);
\draw (0,1.8) -- (1.8,1.8);

\end{tikzpicture}
\caption{Examples of possible shapes with a straight left edge, each with area 10}
\label{fig4}
\end{figure}

To convince $V$ that the area is a rectangle, $P$ needs to perform the ``second flood''. The second flood starts at the bottom-right corner and goes into the cells already visited by the ``first flood'' in the opposite direction from the first flood --- originally the flood can only go upwards (i.e. $P$ can only select the neighbor to the top), then right after it just filled all cells along the right edge, the flood changes direction and can only go leftwards (i.e. $P$ can only select the neighbor to the left).

Formally, $P$ starts at a bottom-right corner of the rectangle and replaces a \mybox{1} with a \mybox{2}. $P$ sets $r=0$ and shows it to $V$. In each step, $P$ secretly chooses whether to add 1 to $r$ or not, then shows $V$ that $r \neq 2$. If $r=0$ (resp. $r=1$), $P$ selects a cell with a \mybox{2} and changes the card on its neighbor to the top (resp. to the left) from a \mybox{1} to a \mybox{2}. $P$ performs this for $t-1$ steps to change all \mybox{1}s in the rectangle to \mybox{2}s.

After the second flood, $P$ have shown that the area also has a straight right edge. This is sufficient to convince $V$ that the area is a rectangle with size $t$ (see the proof of Lemma \ref{lem2} for the full proof of perfect soundness).

In the next section, we will show the necessary subprotocols that enable us to formalize this idea into an actual protocol.

\section{Subprotocols}
\subsection{Addition Protocol for $\mathbb{Z}/3\mathbb{Z}$} \label{add}
We use a sequence of three consecutive cards to encode each integer in $\mathbb{Z}/3\mathbb{Z}$. Namely, we use \mybox{1}\mybox{0}\mybox{0}, \mybox{0}\mybox{1}\mybox{0}, and \mybox{0}\mybox{0}\mybox{1} to encode 0, 1, and 2, respectively.

Suppose we have sequences $R$ and $S$ encoding integers $r$ and $s$ in $\mathbb{Z}/3\mathbb{Z}$, respectively. This protocol, developed by Shinagawa et al. \cite{polygon}, computes the sum $r+s$ without revealing $r$ or $s$.

\begin{enumerate}
	\item Swap the two rightmost cards of $S$. This modified sequence, called $S'$, now encodes $-s$ (mod 3).
	\item Construct a $2 \times 3$ matrix $M$ by placing $S'$ in Row 1 and $R$ in Row 2.
	\item Apply the pile-shifting shuffle to $M$. Note that Row 1 and Row 2 of $M$ now encode $-s+x$ (mod 3) and $r+x$ (mod 3), respectively, for some uniformly random $x \in \mathbb{Z}/3\mathbb{Z}$.
	\item Turn over all cards in Row 1 of $M$. Locate the position of a \mybox{1}. Shift the columns of $M$ cyclically such that this \mybox{1} moves to Column 1.
	\item The sequence in Row 2 of $M$ now encodes $(r+x)-(-s+x) \equiv r+s$ (mod 3) as desired.
\end{enumerate}

Note that Step 4 of this protocol guarantees that $S$ is in a correct format (having one \mybox{1} and two \hbox{\mybox{0}s}). In each step of the flooding protocol, $P$ secretly selects $s \in \{0,1\}$ and places $S$ accordingly. Then, $P$ reveals the rightmost card of $S$ that it is a \mybox{0} to show $V$ that $s \neq 2$. Similarly, after computing the sum $r+s$, $P$ reveals the rightmost card of the resulting sequence to show $V$ that $r+s \neq 2$.

\subsection{Neighbor Selection Protocol} \label{select}
In $\mathbb{Z}/2\mathbb{Z}$, we use \mybox{1}\mybox{0} and \mybox{0}\mybox{1} to encode 0 and 1, respectively. Suppose we have two face-down cards $c_0$ and $c_1$, and a sequence $R$ encoding an integer $r \in \mathbb{Z}/2\mathbb{Z}$. We want to select a card $c_r$ to use in other operations without revealing $r$, and also put $c_0$ and $c_1$ back to where they came from.

We can do so by applying the chosen cut protocol for two cards. However, in Step 1.b, we instead place a sequence $R$ in Row 2 (without revealing $R$). Also, at the end of the chosen cut protocol, $M$ is reverted to its original state, so we can put $c_0$ and $c_1$ back to where they came from.

In each step of the flooding protocol, after showing that $r \neq 2$, $P$ picks only the two leftmost cards of a sequence encoding $r$ in $\mathbb{Z}/3\mathbb{Z}$. This truncated sequence encodes $r$ in $\mathbb{Z}/2\mathbb{Z}$ as desired. During the first flood, $P$ chooses the cards on the neighbor to the bottom and to the right of the selected cell as $c_0$ and $c_1$, respectively; during the second flood, $P$ chooses the cards on the neighbor to the top and to the left of the selected cell as $c_0$ and $c_1$, respectively.

\section{Formal Steps of the Flooding Protocol}
Similarly to the sea formation protocol, we first publicly place a face-down \mybox{0} on every cell in the Shikaku grid, and also place face-down $\mybox{-1}$s around the grid. We now have an $(m+2) \times (n+2)$ matrix of cards (see Figure \ref{fig3}).

Let $h_i = a'_i-a_i+1$ be the height of a rectangle $Z_i$ ($i \in \{2,3,...,k+1\}$). To verify that $Z_i$ is a rectangle with area $p_i$ and also contains a cell with the number $p_i$, $P$ performs the following two phases: the first flood and the second flood.

\subsection{First Flood}
\begin{enumerate}
	\item $P$ applies the chosen cut protocol for $(m+2)(n+2)$ cards to select a card on the top-left corner cell of $Z_i$.
	\item $P$ reveals the selected card to $V$ that it is a \mybox{0} (otherwise $V$ rejects) and then replaces it with a \mybox{1}.
	\item $P$ publicly constructs a sequence $R$ of three cards encoding an integer $r=0$.
	\item $P$ repeatedly performs the following steps for $p_i-1$ iterations.
	\begin{enumerate}
		\item $P$ secretly constructs a sequence $S$ of three cards encoding an integer $s \in \{0,1\}$. If this is the $h_i$-th iteration, $P$ must choose $s=1$; otherwise, $P$ must choose $s=0$.
		\item $P$ reveals the rightmost card of $S$ to $V$ that it is a \mybox{0} to show that $s \neq 2$ (otherwise $V$ rejects).
		\item $P$ applies the addition protocol to compute $r+s$ and reveals the rightmost card of the resulting sequence to $V$ that it is a \mybox{0} to show that $r+s \neq 2$ (otherwise $V$ rejects). From now on, set $r := r+s$.
		\item $P$ applies the chosen cut protocol for $(m+2)(n+2)$ cards to select a \mybox{1} he/she wants from the Shikaku grid. If this is during the first $h_i-1$ iterations, $P$ must choose the bottommost \mybox{1}; otherwise, $P$ may choose any card that is the rightmost \mybox{1} in its row and is not located in the rightmost column of $Z_i$.
		\item $P$ reveals the selected card to $V$ that it is a \mybox{1} (otherwise $V$ rejects).
		\item $P$ chooses the neighbors to the bottom and to the right of the selected card as $c_0$ and $c_1$, respectively, and applies the neighbor selection protocol to select a card $c_r$ (using the two leftmost cards of a sequence encoding $r$ as inputs).
		\item $P$ reveals the selected neighbor to $V$ that it is a \mybox{0} (otherwise $V$ rejects) and then replaces it with a \mybox{1}.
	\end{enumerate}
\end{enumerate}

After the first flood, all cards on the cells in $Z_i$ are now changed to \mybox{1}s.

\subsection{Second Flood}
\begin{enumerate}
	\item $P$ applies the chosen cut protocol for $(m+2)(n+2)$ cards to select a card on the bottom-right corner cell of $Z_i$.
	\item $P$ reveals the selected card to $V$ that it is a \mybox{1} (otherwise $V$ rejects) and then replaces it with an \mybox{$i$}.
	\item $P$ publicly constructs a sequence $R$ of three cards encoding an integer $r=0$.
	\item $P$ repeatedly performs the following steps for $p_i-1$ iterations.
	\begin{enumerate}
		\item $P$ secretly constructs a sequence $S$ of three cards encoding an integer $s \in \{0,1\}$. If this is the $h_i$-th iteration, $P$ must choose $s=1$; otherwise, $P$ must choose $s=0$.
		\item $P$ reveals the rightmost card of $S$ to $V$ that it is a \mybox{0} to show that $s \neq 2$ (otherwise $V$ rejects).
		\item $P$ applies the addition protocol to compute $r+s$ and reveals the rightmost card of the resulting sequence to $V$ that it is a \mybox{0} to show that $r+s \neq 2$ (otherwise $V$ rejects). From now on, set $r := r+s$.
		\item $P$ applies the chosen cut protocol for $(m+2)(n+2)$ cards to select an \mybox{$i$} he/she wants from the Shikaku grid. If this is during the first $h_i-1$ iterations, $P$ must choose the topmost \mybox{$i$}; otherwise, $P$ may choose any card that is the leftmost \mybox{$i$} in its row and is not located in the leftmost column of $Z_i$.
		\item $P$ reveals the selected card to $V$ that it is an \mybox{$i$} (otherwise $V$ rejects).
		\item $P$ chooses the neighbors to the top and to the left of the selected card as $c_0$ and $c_1$, respectively, and applies the neighbor selection protocol to select a card $c_r$ (using the two leftmost cards of a sequence encoding $r$ as inputs).
		\item $P$ reveals the selected neighbor to $V$ that it is a \mybox{1} (otherwise $V$ rejects) and then replaces it with an \mybox{$i$}.
	\end{enumerate}
\end{enumerate}

After the second flood, all cards on the cells in $Z_i$ are now changed to \mybox{$i$}s. Finally, $P$ turns over a card on the cell with the number $p_i$ to show that it is an \mybox{$i$}, i.e. $Z_i$ contains the cell with the number $p_i$ (otherwise $V$ rejects).

$P$ performs the above two phases for every $i \in \{2,3,...,k+1\}$. If all verification steps pass, then $V$ accepts.

The number of cards used in the flooding protocol is $\Theta(mn)$, which is much lower than the brute force protocol.

\section{Proof of Security}
We will prove the perfect completeness, perfect soundness, and zero-knowledge properties of the flooding protocol.

\begin{lemma}[Perfect Completeness] \label{lem1}
If $P$ knows a solution of the Shikaku puzzle, then $V$ always accepts.
\end{lemma}

\begin{proof}
Suppose that $P$ knows a solution of the puzzle. Consider the verification of each $Z_i$.

In the first flood, during the first $h_i-1$ iterations $P$ chooses $s=0$ and chooses the bottommost \mybox{1}, so the area of \mybox{1}s expands downwards by one cell. After $h_i-1$ iterations, all cards along the left edge of $Z_i$ have been changed to \mybox{1}s. In the $h_i$-th iteration, $P$ chooses $s=1$ and chooses any \mybox{1}, so the flood direction is changed to rightwards and the area of \hbox{\mybox{1}s} expands by one cell. After that, in each iteration $P$ chooses $s=0$ and chooses any card that is the rightmost \mybox{1} in its row and is not located in the rightmost column of $Z_i$, so the area of \mybox{1}s expands by one cell inside $Z_i$. Therefore, at the end of the first flood, all cards in $Z_i$ has been changed to \mybox{1}s.

Analogously, in the second flood, during the first $h_i-1$ iterations $P$ chooses $s=0$ and chooses the topmost \mybox{$i$}, so the area of \mybox{$i$}s expands upwards by one cell. After $h_i-1$ iterations, all cards along the right edge of $Z_i$ have been changed to \mybox{$i$}s. In the $h_i$-th iteration, $P$ chooses $s=1$ and chooses any \mybox{$i$}, so the flood direction is changed to leftwards and the area of \mybox{$i$}s expands by one cell. After that, in each iteration $P$ chooses $s=0$ and chooses any card that is the leftmost \mybox{$i$} in its row and is not located in the leftmost column of $Z_i$, so the area of \mybox{$i$}s expands by one cell inside $Z_i$. Therefore, at the end of the first flood, all cards in $Z_i$ has been changed to \mybox{$i$}s, thus a card on the cell containing the number $p_i$ must also be an \mybox{$i$}.

Since the verification passes for every $Z_i$, $V$ always accepts.
\end{proof}

\begin{lemma}[Perfect Soundness] \label{lem2}
If $P$ does not know a solution of the Shikaku puzzle, then $V$ always rejects.
\end{lemma}

\begin{proof}
We will prove the contrapositive of this statement. Suppose that $V$ accepts, meaning that the flooding protocol passes for every $Z_i$. We will prove that $P$ must know a solution.

First, note that the chosen cut protocol in Section \ref{chosen} and the addition protocol in Section \ref{add} guarantee that the inputs from $P$ must be in a correct format. Consider the verification of $Z_i$. Suppose that the first flood goes downwards for $h-1$ steps before changing direction to rightwards. The area that contains \mybox{1}s after the first flood must have a straight left edge with height $h$, and have a shape like $h$ horizontal bars placing on top of each other. Let $\ell_1,\ell_2,...,\ell_h$ be the length of these bars from top to bottom. For example, in Figure \ref{fig5} we have $h=4$, $\ell_1=2$, $\ell_2=4$, $\ell_3=1$, and $\ell_4=3$.

\begin{figure}[h]
\centering
\begin{tikzpicture}
\draw (0,0) -- (0,2.4);
\draw (0.6,0) -- (0.6,2.4);
\draw (1.2,0) -- (1.2,0.6);
\draw (1.2,1.2) -- (1.2,2.4);
\draw (1.8,0) -- (1.8,0.6);
\draw (1.8,1.2) -- (1.8,1.8);
\draw (2.4,1.2) -- (2.4,1.8);

\draw (0,0) -- (1.8,0);
\draw (0,0.6) -- (1.8,0.6);
\draw (0,1.2) -- (2.4,1.2);
\draw (0,1.8) -- (2.4,1.8);
\draw (0,2.4) -- (1.2,2.4);
\end{tikzpicture}
\caption{An example of a possible shape with a straight left edge}
\label{fig5}
\end{figure}
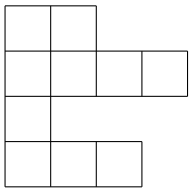

Since all $p_i$ \mybox{1}s in this area have been replaced by \mybox{$i$}s after the second flood, all cells in the area must be reachable from the starting point of the second flood by moving only upwards or leftwards. Thus, the only possible starting point of the second flood is the rightmost cell of the bottommost bar (the one with length $\ell_h$).

Moreover, for any $j<h$, we must have $\ell_j \leq \ell_h$ (otherwise there is a cell in the $j$-th bar which is located to the right of the starting point and thus not reachable by the second flood) However, if $\ell_j < \ell_h$, the second flood cannot go directly from the starting point to the $j$-th bar by only moving upwards; it has to change direction at least twice, a contradiction since the flood can change direction at most once. Therefore, we must have $\ell_j = \ell_h$ for every $j \in \{1,2,...,h-1\}$, which means the area must be a rectangle.

Therefore, $Z_i$ is a rectangle with area $p_i$ that contains a cell with the number $p_i$ for every $i \in \{2,3,...,k+1\}$. Since any two rectangles do not overlap, and $p_2+p_3+...+p_{k+1} = mn$, they must be a partition of the grid. Hence, we can conclude that $P$ knows a valid solution of the puzzle.
\end{proof}

\begin{lemma}[Zero-Knowledge] \label{lem3}
During the verification phase, $V$ learns nothing about $P$'s solution of the Shikaku puzzle.
\end{lemma}

\begin{proof}
To prove the zero-knowledge property, it is sufficient to show that all distributions of cards that are turned face-up can be simulated by a simulator $S$ that does not know $P$'s solution.

\begin{itemize}
	\item In Steps 3 and 6 of the chosen cut protocol in Section \ref{chosen}, the \mybox{1} has an equal probability to be at any of the $q$ positions, so this step can be simulated by $S$.
	\item In Step 4 of the addition protocol in Section \ref{add}, the \mybox{1} has an equal probability to be at any of the three positions, so this step can be simulated by $S$.
	\item In the flooding protocol, during the verification of each $Z_i$, there is only one deterministic pattern of the cards that are turned face-up. This pattern solely depends on $p_i$, which is public information, so the whole protocol can be simulated by $S$.
\end{itemize}
\end{proof}

\section{Future Work}
We developed a physical ZKP protocol with perfect completeness and soundness for Shikaku using $\Theta(mn)$ cards. Most importantly, we also developed a general technique to physically verify a rectangle-shaped area with a certain size in a rectangular grid.

A possible future work is to develop physical ZKP protocols to verify other geometric shapes or other puzzles with constraints related to shapes (e.g. Shakashaka). Another interesting future work is to develop an equivalent protocol for Shikaku that can be implemented using a deck of all different cards (like the one for Sudoku \cite{sudoku2}).


\begin{thebibliography}{99}
	\bibitem{akari} X. Bultel, J. Dreier, J.-G. Dumas and P. Lafourcade. Physical Zero-Knowledge Proofs for Akari, Takuzu, Kakuro and KenKen. In \textit{Proceedings of the 8th International Conference on Fun with Algorithms (FUN)}, pp. 8:1--8:20 (2016).
	\bibitem{makaro} X. Bultel, J. Dreier, J.-G. Dumas, P. Lafourcade, D. Miyahara, T. Mizuki, A. Nagao, T. Sasaki, K. Shinagawa and H. Sone. Physical Zero-Knowledge Proof for Makaro. In \textit{Proceedings of the 20th International Symposium on Stabilization, Safety, and Security of Distributed Systems (SSS)}, pp. 111--125 (2018).
	\bibitem{nonogram} Y.-F. Chien and W.-K. Hon. Cryptographic and Physical Zero-Knowledge Proof: From Sudoku to Nonogram. In \textit{Proceedings of the 5th International Conference on Fun with Algorithms (FUN)}, pp. 102--112 (2010).
	\bibitem{norinori} J.-G. Dumas, P. Lafourcade, D. Miyahara, T. Mizuki, T. Sasaki and H. Sone. Interactive Physical Zero-Knowledge Proof for Norinori. In \textit{Proceedings of the 25th International Computing and Combinatorics Conference (COCOON)}, pp. 166--177 (2019).
	\bibitem{zkp} O. Goldreich, S. Micali and A. Wigderson. Proofs that yield nothing but their validity and a methodology of cryptographic protocol design. \textit{Journal of the ACM}, 38(3): 691--729 (1991).
	\bibitem{zkp0} S. Goldwasser, S. Micali and C. Rackoff. The knowledge complexity of interactive proof systems. \textit{SIAM Journal on Computing}, 18(1): 186--208 (1989).
	\bibitem{google} Google Play: Shikaku. \url{https://play.google.com/store/search?q=Shikaku&c=apps}
	\bibitem{sudoku0} R. Gradwohl, M. Naor, B. Pinkas and G.N. Rothblum. Cryptographic and Physical Zero-Knowledge Proof Systems for Solutions of Sudoku Puzzles. \textit{Theory of Computing Systems}, 44(2): 245--268 (2009).
	\bibitem{crypta} R. Isuzugawa, D. Miyahara and T. Mizuki. Zero-Knowledge Proof Protocol for Cryptarithmetic Using Dihedral Cards. In \textit{Proceedings of the 19th International Conference on Unconventional Computation and Natural Computation (UCNC)}, pp. 51--67 (2021).
	\bibitem{koch} A. Koch and S. Walzer. Foundations for Actively Secure Card-Based Cryptography. In \textit{Proceedings of the 10th International Conference on Fun with Algorithms (FUN)}, pp. 17:1--17:23 (2020).
	\bibitem{slitherlink} P. Lafourcade, D. Miyahara, T. Mizuki, L. Robert, T. Sasaki and H. Sone. How to construct physical zero-knowledge proofs for puzzles with a ``single loop'' condition. \textit{Theoretical Computer Science}, 888: 41--55 (2021).
	\bibitem{takuzu} D. Miyahara, L. Robert, P. Lafourcade, S. Takeshige, T. Mizuki, K. Shinagawa, A. Nagao and H. Sone. Card-Based ZKP Protocols for Takuzu and Juosan. In \textit{Proceedings of the 10th International Conference on Fun with Algorithms (FUN)}, pp. 20:1--20:21 (2020).
	\bibitem{kakuro} D. Miyahara, T. Sasaki, T. Mizuki and H. Sone. Card-Based Physical Zero-Knowledge Proof for Kakuro. \textit{IEICE Transactions on Fundamentals of Electronics, Communications and Computer Sciences}, E102.A(9): 1072--1078 (2019).
	\bibitem{nurikabe} L. Robert, D. Miyahara, P. Lafourcade and T. Mizuki. Interactive Physical ZKP for Connectivity: Applications to Nurikabe and Hitori. In \textit{Proceedings of the 17th Conference on Computability in Europe (CiE)}, pp. 373--384 (2021).
	\bibitem{suguru} L. Robert, D. Miyahara, P. Lafourcade and T. Mizuki. Physical Zero-Knowledge Proof for Suguru Puzzle. In \textit{Proceedings of the 22nd International Symposium on Stabilization, Safety, and Security of Distributed Systems (SSS)}, pp. 235--247 (2020).
	\bibitem{nonogram2} S. Ruangwises. An Improved Physical ZKP for Nonogram. In \textit{Proceedings of the 15th Annual International Conference on Combinatorial Optimization and Applications (COCOA)}, pp. 262--272 (2021).
	\bibitem{sudoku2} S. Ruangwises. Two Standard Decks of Playing Cards Are Sufficient for a ZKP for Sudoku. In \textit{Proceedings of the 27th International Computing and Combinatorics Conference (COCOON)}, pp. 631--642 (2021).
	\bibitem{numberlink} S. Ruangwises and T. Itoh. Physical Zero-Knowledge Proof for Numberlink Puzzle and $k$ Vertex-Disjoint Paths Problem. \textit{New Generation Computing}, 39(1): 3--17 (2021).
	\bibitem{ripple} S. Ruangwises and T. Itoh. Physical Zero-Knowledge Proof for Ripple Effect. \textit{Theoretical Computer Science}, 895: 115--123 (2021).
	\bibitem{bridges} S. Ruangwises and T. Itoh. Physical ZKP for Connected Spanning Subgraph: Applications to Bridges Puzzle and Other Problems. In \textit{Proceedings of the 19th International Conference on Unconventional Computation and Natural Computation (UCNC)}, pp. 149--163 (2021).
	\bibitem{sudoku} T. Sasaki, D. Miyahara, T. Mizuki and H. Sone. Efficient card-based zero-knowledge proof for Sudoku. \textit{Theoretical Computer Science}, 839: 135--142 (2020).
	\bibitem{polygon} K. Shinagawa, T. Mizuki, J.C.N. Schuldt, K. Nuida, N. Kanayama, T. Nishide, G. Hanaoka and E. Okamoto. Card-Based Protocols Using Regular Polygon Cards. \textit{IEICE Transactions on Fundamentals of Electronics, Communications and Computer Sciences}, E100.A(9): 1900--1909 (2017).
	\bibitem{np} Y. Takenaga, S. Aoyagi, S. Iwata and T. Kasai. Shikaku and Ripple Effect are NP-Complete. \textit{Congressus Numerantium}, 216: 119--127 (2013).
	\bibitem{hindu} I. Ueda, D. Miyahara, A. Nishimura, Y. Hayashi, T. Mizuki and H. Sone. Secure implementations of a random bisection cut. \textit{International Journal of Information Security}, 19(4): 445--452 (2020).
\end{thebibliography}
\end{document}